\documentclass[12pt]{article}
\usepackage[latin1]{inputenc}
\usepackage[british]{babel}
\usepackage{cmap}
\usepackage{lmodern}

\usepackage{amssymb, amsmath, amsthm}
\usepackage[a4paper,top=25mm,bottom=25mm,left=25mm,right=25mm]{geometry}
\usepackage{ragged2e}

\usepackage{authblk} % for headings
\usepackage{pifont}
\usepackage{graphicx}
\usepackage[usenames,dvipsnames,svgnames,table]{xcolor}
\usepackage[figuresright]{rotating}
\usepackage{xtab} % tackle the long tables
\usepackage{longtable} % tackle the long tables
\usepackage{multirow}
\usepackage{footnote}
\usepackage[stable]{footmisc}
\usepackage{chngpage} % allows for temporary adjustment of side margins
\usepackage{pdflscape} % landscape environment
\usepackage[nottoc,notlot,notlof]{tocbibind} % includes everything in ToC

\usepackage{pgfplots}
\pgfplotsset{compat=1.14}
\pgfplotsset{every tick label/.append style={font=\footnotesize}}
\usepgfplotslibrary{fillbetween}
\usepackage{setspace}
\makesavenoteenv{tabular}
\usepackage{tabularx}
\usepackage{booktabs}
\usepackage[flushleft]{threeparttable}
\usepackage[referable]{threeparttablex} % footnotes in tabu
\newcolumntype{R}{>{\raggedleft\arraybackslash}X}
\newcolumntype{L}{>{\raggedright\arraybackslash}X}
\newcolumntype{C}{>{\centering\arraybackslash}X}
\newcolumntype{M}[1]{>{\centering\arraybackslash}m{#1}}
\newcolumntype{K}{>{\columncolor{gray!20}}C}
\newcolumntype{k}{>{\columncolor{gray!20}}c}

\newlength{\tablen}

\usepackage{dcolumn} % alignment to decimal points
\newcolumntype{.}{D{.}{.}{-1}}

\usepackage{tikz}
\usetikzlibrary{arrows, calc, matrix, patterns, positioning, trees}
\usepackage[semicolon]{natbib}
\usepackage[hyphens]{url}
\usepackage{hyperref} % [hidelinks]
\hypersetup{
  colorlinks   = true,    % Colours links instead of ugly boxes
  urlcolor     = blue,    % Colour for external hyperlinks
  linkcolor    = blue,    % Colour of internal links
  citecolor    = ForestGreen      % Colour of citations
}
\usepackage{microtype}
\usepackage[justification=centerfirst]{caption} % [justification=centering]

\usepackage[labelformat=simple]{subcaption}

\DeclareCaptionLabelFormat{parenthesis}{(#2)}
\makeatletter
\renewcommand\p@subfigure{\arabic{figure}.}
\makeatother

\DeclareCaptionLabelFormat{parenthesis}{(#2)}
\makeatletter
\renewcommand\p@subtable{\arabic{table}.}
\makeatother

\usepackage{enumitem}
\setlist[itemize]{leftmargin=2.5\parindent}
\setlist[enumerate]{leftmargin=2.5\parindent}

\def\addlegendimage{\csname pgfplots@addlegendimage\endcsname}

\theoremstyle{plain}
\newtheorem{proposition}{Proposition}%[section]

\theoremstyle{definition}
\newtheorem{axiom}{Axiom} %[section]

\newtheorem{definition}{Definition}%[section]
\newtheorem{example}{Example}[section]

\theoremstyle{remark}

\def\keywords{\vspace{.5em} % Add keywords
{\noindent \textit{Keywords}: }}

\def\JEL{\vspace{.5em} % Add keywords
{\noindent \textbf{\emph{JEL} classification number}: }}

\def\AMS{\vspace{.5em} % Add keywords
{\noindent \textbf{\emph{MSC} class}: }}

\author{
\href{https://sites.google.com/view/laszlocsato}{L\'aszl\'o Csat\'o}\thanks{~Corresponding author. E-mail: \emph{laszlo.csato@sztaki.hu}}}
\affil{Institute for Computer Science and Control (SZTAKI) \\
Laboratory on Engineering and Management Intelligence, Research Group of Operations Research and Decision Systems}
\affil{Corvinus University of Budapest (BCE) \\
Department of Operations Research and Actuarial Sciences}
\affil{Budapest, Hungary}

\title{Coronavirus and sports leagues: obtaining \\ a fair ranking when the season cannot resume}
\date{\today}

\def\Dedication{ % Add keywords
{\noindent $\mathfrak{Sollen}$ $\mathfrak{wir}$ $\mathfrak{aus}$ $\mathfrak{der}$ $\mathfrak{Geschichte}$ $\mathfrak{lernen}$, $\mathfrak{so}$ $\mathfrak{m\ddot{u}ssen}$ $\mathfrak{wir}$ $\mathfrak{die}$ $\mathfrak{Dinge}$, $\mathfrak{weiche}$ $\mathfrak{sich}$ $\mathfrak{wirklich}$ $\mathfrak{zugetragen}$ $\mathfrak{haben}$, $\mathfrak{doch}$ $\mathfrak{auch}$ $\mathfrak{f\ddot{u}r}$ $\mathfrak{die}$ $\mathfrak{Folge}$ $\mathfrak{als}$ $\mathfrak{m\ddot{o}glich}$ $\mathfrak{ansehen}$.}\footnote{~``\emph{If we are to learn from history, we must look upon things which have actually happened as also possible in the future.}'' (Source: Carl von Clausewitz: \emph{On War}, Book 8, Chapter 8---Limited Object---Defence, translated by Colonel James John Graham, London, N. Tr\"ubner, 1873. \url{http://clausewitz.com/readings/OnWar1873/TOC.htm})}
\flushright
\noindent (Carl von Clausewitz: \emph{Vom Kriege})

\vspace{1cm} 
\justify }

\begin{document}
%\newgeometry{top=20mm,bottom=20mm,left=25mm,right=25mm}

\maketitle
\thispagestyle{empty}
\Dedication

\begin{abstract}
\noindent
Many sports leagues are played in a tightly scheduled round-robin format, leaving a limited time window to postpone matches. If the season cannot resume due to an external shock such as the coronavirus pandemic in 2020, the ranking of the teams becomes non-trivial: it is necessary to account for schedule imbalances and possibly for the different number of matches played. First in the literature, we identify a set of desired axioms for ranking in these incomplete tournaments. It is verified that the generalised row sum, a parametric family of scoring rules, satisfies all of them. In particular, the well-established least squares method maximizes the influence of the strength of opponents on the ranking.
Our approach is applied for six major premier European soccer competitions, where the rankings are found to be robust concerning the weight of adjustment for the strength of the opponents. Some methodologically simpler alternative policies are also discussed.
%Since disregarding the majority of the matches without promotion, relegation, and qualification for international cups would be unfair and unjustified, some simpler alternative policies are also discussed.


\AMS{15A06, 91B14}
% Linear equations (linear algebraic aspects)
% Social choice

\JEL{C44, Z20}
% Operations Research, Statistical Decision Theory
% Sport Economics: General
\end{abstract}

\clearpage
%\newgeometry{top=25mm,bottom=25mm,left=25mm,right=25mm}

\section{Introduction} \label{Sec1}

The pandemic of coronavirus disease 2019 (COVID-19) stopped almost all sports leagues around the world. Some of them never resumed: the German premier men's handball league, the \href{https://en.wikipedia.org/wiki/2019\%E2\%80\%9320_Handball-Bundesliga}{2019/20 Handball-Bundesliga}, was cancelled on 21 April 2020 \citep{Web24_2020}; the Dutch soccer league, the \href{https://en.wikipedia.org/wiki/2019\%E2\%80\%9320_Eredivisie}{2019/20 Eredivisie} ended with immediate effect on 24 April 2020 \citep{SkySports2020}; and the French Prime Minister announced on 28 April 2020 that the 2019/20 sporting season is over \citep{BBC2020a}.
Therefore, several organisers have faced an unenviable dilemma: how to decide the final ranking in the 2019/20 season? Should all results be abandoned? Should the current standing be frozen? Should a reasonable subset of the matches be considered?
Since the sports industry is a billion dollars business, the answer has huge financial consequences as promotion and relegation, qualification for international cups, and the allocation of broadcasting revenue \citep{BergantinosMoreno-Ternero2020a} all depend on the league ranking.

%Some executives, for instance, the president of the French soccer club Olympique Lyonnais, promotes disregarding the whole season without announcing the champion, as well as deleting promotion and relegation \citep{Hernandez2020}. This implies that the next season of international cups---in European soccer, the UEFA Champions League and UEFA Europa League---should start with the same participants as in 2019/20. The solution obviously favours some teams. Olympique Lyonnais has been only the seventh after playing 28 matches from a total of 38, with a small chance to qualify for the UEFA Champions League on the field.

For instance, the Hungarian Handball Federation cancelled all results in the 2019/20 season \citep{Stregspiller2020}. The implications for the \href{https://en.wikipedia.org/wiki/2019\%E2\%80\%9320_Nemzeti_Bajnoks\%C3\%A1g_I_(women\%27s_handball)}{2019/20 Hungarian women's handball league} were quite dramatic.
With eight matches left to play from the total of 26, the runner-up Si\'ofok had an advantage of two points---the reward of a win---ahead of the third-placed Ferencv\'aros. However, based on the result of the previous 2018/19 season, Ferencv\'aros obtained the second slot in the 2020/21 EHF Champions League for Hungary. Thus Si\'ofok, which finished third in the 2018/19 season, can enter only the less lucrative and prestigious EHF European League (called EHF Cup until the 2019/20 season), a competition it already won in 2019. Unsurprisingly, the strange decision ignited some controversy \citep{Stregspiller2020}, especially because Gy\H{o}ri Audi ETO KC---the Hungarian club that won the last three editions of the \href{https://en.wikipedia.org/wiki/Women\%27s_EHF_Champions_League}{Women's EHF Champions League}, the most prestigious club tournament in women's handball---was defeated by Si\'ofok on 22 January 2020, the first loss of Gy\H{o}r after more than two years.

%However, as the matches have usually been scheduled between the autumn of 2019 and the spring of 2020, more than of half of them has already been played.

On the other hand, the best 36 German men's handball clubs voted by a large majority to evaluate the 2019/20 season according to the so-called quotient rule, the number of points scored per game \citep{Web24_2020}. In addition, the top two clubs were promoted from the second division but there was no relegation.
Even though F\"uchse Berlin dropped from the fifth to the sixth place due to playing one match more despite its one point advantage over Rhein-Neckar L\"owen, and thus can play in the EHF European League only due to the withdrawal of the fourth club TSV Hannover-Burgdorf, its managing director acknowledged that there is no fair solution in this situation.

The current paper aims to propose a fair ranking in similar incomplete round-robin tournaments, where the main challenge is how to take into account the varying strength of the opponents. The analysis is based on preserving much of the theoretical properties of the standard ranking in round-robin leagues. Therefore, the novelty of our research resides in its \emph{axiomatic approach}, which has not been followed in this context to our knowledge.
On the other hand, the social literature has extensively discussed similar issues \citep{ChebotarevShamis1997a, ChebotarevShamis1998a, Gonzalez-DiazHendrickxLohmann2014, Kitti2016, Rubinstein1980, SlutzkiVolij2005, SlutzkiVolij2006}. \citet{VaziriDabadghaoYihMorin2018} consider the conditions guaranteeing the fairness and comprehensiveness of sports ranking methods, while \citet{Berker2014} analyses the theoretical properties of tie-breaking rules. Axioms for penalty shootout rules are also widely studied \citep{BramsIsmail2018, LambersSpieksma2020a, Csato2021c}, and \citet{Csato2021a} presents tournament design from the perspective of fairness and incentive compatibility.

A family of scoring rules for preference aggregation, the generalised row sum method, is found to satisfy all axioms that remain compatible with each other. It depends on a parameter responsible for the adjustment due to the varying strength of the opponents. Maximising the influence of the opponents leads to the least squares method, a well-established procedure in social choice theory and statistics. According to our computations for six major European soccer leagues, the current ranking---after correcting for the number of matches played as in the 2019/20 Handball-Bundesliga---remains robust to variation in the strength of the opponents.

The problem of predicting the winner of an interrupted game has a long history since the correspondence of \emph{Blaise Pascal} and \emph{Pierre de Fermat} \citep[Appendix: The Pascal--Fermat Correspondence of 1654]{Weisberg2014}. There exist many statistical techniques to determine the number of points that a team would have won on average in the rest of the season. However, sports administrators are usually not keen to grab such sophisticated mathematical models, even though cricket has adopted the \href{https://en.wikipedia.org/wiki/Duckworth\%E2\%80\%93Lewis\%E2\%80\%93Stern_method}{Duckworth--Lewis--Stern method} \citep{DuckworthLewis1998, Stern2009, Wright2009}.
Our paper is also not the first attempt to produce a final ranking in incomplete sports tournaments. For example, the French mathematician \emph{Julien Guyon} suggests an Elo-based method in the mainstream media \citep{Guyon2020b, Guyon2020c}, while, inspired by \citet{Landau1914}, \citet{LambersSpieksma2020b} apply an eigenvector-based approach. \citet{GorgiKoopmanLit2020} propose a procedure to predict the remaining non-played matches through a paired-comparison model, and \citet{VanEetveldeHvattumLey2021} introduce a stochastic tool to derive the probabilities for the various possible final rankings. \citet{BeggsBondEmmondsJones2019} analyse historical partial standings from English soccer to evaluate their predictive ``power''.

However, none of the above works have used axiomatic arguments for their proposals. This inspired us to look for a simple solution that retains most properties of the usual ranking in round-robin tournaments as the majority of forecasting models do not guarantee these axioms. Our proposal is based exclusively on the number of points as well as the structure of the matches already played, and has the following advantages:
\begin{itemize}
\item
It takes into consideration all results from the unfinished season;
\item
It accounts for the different number of matches played as well as for the fact that the teams played against different sets of opponents with varying strengths;
\item
It does not depend on the form, home-away pattern, injuries, or results in the preceding season(s);
\item
It seldom requires further tie-breaking criteria \citep{Csato2017c};
\item
It is suitable for any sport and does not call for the estimation of any parameter;
\item
It has a recursive formula where the first iteration still yields an adequate solution, consequently, it requires no specific software, and the ranking can be computed essentially by hand.
\end{itemize}

On the other hand, it should be acknowledged that no reasonable prediction is provided for the final number of points. In our opinion, that is not a serious problem as only the ranking of the teams is required to decide promotion, relegation, and qualification to other tournaments.

The rest of the article proceeds as follows.
Section~\ref{Sec2} studies some properties of the usual ranking in round-robin tournaments and discusses their relevance when all matches cannot be played. The generalised row sum method is introduced in Section~\ref{Sec3} and shown to be compatible with the desired axioms. Section~\ref{Sec4} contains its application for six major European soccer leagues. Section~\ref{Sec5} summarises policy implications, while Section~\ref{Sec6} offers concluding remarks.

\section{Axioms} \label{Sec2}

As the first step, we present some conditions that are satisfied by the ranking obtained from the number of points in a round-robin tournament. This collection aims to assess whether they can be guaranteed if the league cannot resume after some matches are played. Precise mathematical formalization is avoided in order to be more accessible for the general audience.

\begin{axiom}
\emph{Independence of irrelevant matches}:
The relative ranking of two teams should be independent of any matches between other teams.
\end{axiom}

Obviously, the number of points scored by any team does not depend on the results of other teams. However, if head-to-head results are used for tie-breaking, this property does not necessarily hold. \citet{Berker2014} provides an extensive analysis of tie-breaking in round-robin soccer tournaments.

\begin{axiom}
\emph{Self-consistency}:
Two teams should have the same rank if they achieved the same results against teams having the same strength. Furthermore, a team should be ranked strictly higher than another team if one of the following conditions hold:
\begin{itemize}
\item
it achieved better results against opponents having the same strength;
\item
it achieved the same results against stronger opponents;
\item
it achieved better results against stronger opponents.
\end{itemize}
\end{axiom}
Self-consistency imposes some constraints on the ranking of the teams such that the constraints are determined by this ranking itself. \citet[Section~3.2]{Csato2019d} offers a detailed explanation of the axiom.

The exact meaning of ``better results'' against another team can be specified by the decision-maker, however, since the final ranking is usually determined by the number of points scored, a reasonable definition is more points scored. Then having better results implies a higher score, thus the ranking given by the number of points satisfies the first and the third requirements of self-consistency. 
Proving that the second condition also holds is somewhat more complicated.
Consider two teams $i$ and $j$. Assume for contradiction that team $i$ has the same results against stronger opponents compared to team $j$ but team $i$ is ranked weakly below team $j$. In a round-robin tournament, their set of opponents almost coincide but team $i$ is an opponent of team $j$ and vice versa. Consequently, the initial assumption cannot hold because team $i$ faced stronger opponents than team $j$.
% Hence team $i$ and team $j$ should have the same rank because they have the same number of points.

If the tournament is not round-robin, then there exists no ranking method simultaneously satisfying independence of irrelevant matches and self-consistency \citep[Theorem~1]{Csato2019d}. The impossibility result holds under various domain restrictions. Self-consistency is a more important property than independence of irrelevant matches, the latter should be weakened to retain possibility \citep{Csato2019d}.
%An appropriate version is macrovertex independence \citep{Chebotarev1994, Csato2019d} and macrovertex autonomy \citep{Csato2019d}, which can be even extended to bridge set independence and bridge set autonomy \citep{CsatoToth2020}. The family of ranking methods to be proposed later  satisfies these properties.

\begin{axiom}
\emph{Invariance to cycles}:
Assume that team $i$ defeated team $j$, team $j$ defeated team $k$, and team $k$ defeated team $j$. The ranking should be independent of reversing these results such that team $j$ wins against team $i$, team $k$ wins against team $j$, and team $i$ wins against team $k$.
\end{axiom}

%Naturally, the points scored by any team do not change if the results are reversed along such a cycle.
Such cycles are responsible for the basic challenges of ranking. While they certainly cannot be deleted, it would be unjustified to modify the ranking after reversing their direction. Most statistical methods are unlikely to guarantee invariance to cycles.

\begin{axiom}
\emph{Home-away independence}:
The ranking of the teams should not change if the field of all matches is reversed: if the game was played by team $i$ against team $j$ at home, then it is considered as team $i$ against team $j$ away, while the outcome remains unchanged.
\end{axiom}

The number of points does not take into account whether a game is played at home or away (tie-breaking rules in some leagues do). It is a natural requirement in a round-robin tournament where the teams usually play once at home and once away, however, this home-away balance does not hold if the league is suspended.
For example, Rennes, the third team in the 2019/20 French Ligue 1 after 28 rounds, played only at home against the dominating team Paris Saint-Germain, while the fourth Lille lost against Paris Saint-Germain both at home and away. On the other hand, Rennes would have to play five, and Lille would have to play six games away in the remaining 10 rounds.

There are at least two arguments for home-away independence even in incomplete round-robin tournaments. First, according to our knowledge, this feature is not taken into account in any European soccer league even though it would make sense in certain formats. In the 2019/20 season, the leagues in Hungary and Kosovo are organised with 12 teams playing a triple round-robin tournament, hence some teams play against others twice at home and once away, or vice versa. In addition, certain teams play one match fewer at home than other teams. The same format was used in Finland until 2018.
In Northern Ireland and Scotland, the 12 teams play a triple round-robin tournament in the initial stage. After 33 games, the league is split into two sections of six teams each such that every team plays once more all the five teams in their section. Consequently, it is impossible to balance the home-away pattern by guaranteeing that each club plays twice at home and twice away against any other club in its section. While these unbalanced schedules generate interesting research problems, for example, concerning the measurement of competitive balance \citep{Lenten2008}, the number of points---and further tie-breaking rules---are not adjusted for the inherent inequality.

Second, accounting for home advantage requires a statistical estimation of at least one parameter, which will certainly result in debates around the exact methodology and the sample used. As the above examples illustrate, it is improbable that the decision-makers want to instigate such controversies.

\begin{axiom}
\emph{Consistency}:
As the number of rounds played increases, the ranking should converge to the ranking given by the number of points.
In particular, the ranking should coincide with the ranking given by the number of points when the competition is finished.
\end{axiom}

Consistency is perhaps the most essential property of ranking the teams in a suspended league. The number of points has no competitive alternative in the real-world as the primary ranking criterion. It is also important to use a procedure that satisfies consistency due to its innate characteristics, not merely because it is defined separately for incomplete and complete round-robin tournaments. For example, this condition does not hold for the eigenvector-based method proposed in \citet{LambersSpieksma2020b}.

To summarise, five properties of the common ranking method in round-robin tournaments have been presented. Independence of irrelevant matches cannot be expected to hold if the league is suspended. Thus we are seeking a method satisfying self-consistency, invariance to cycles, home-away independence, and consistency. 

\section{A reasonable family of ranking methods} \label{Sec3}

Our proposal for ranking in incomplete round-robin tournaments requires two inputs: the \emph{points vector} $\mathbf{p}$ and the symmetric \emph{matches matrix} $\mathbf{M}$. In particular, $p_i$ is the number of points for team $i$, whereas $m_{ij}$ equals the number of matches between teams $i$ and $j$.

The methods are based on the \emph{normalised points vector} $\mathbf{s}$ with its entries summing up to zero. This is straightforward in sports where a win plus a loss is equivalent to two draws since wins can be awarded by $\alpha > 0$, draws by zero, and losses by $-\alpha < 0$.
But some sports do not follow this symmetric setting, mainly because of the three points rule, which has become the standard in soccer: wins earn three points, draws earn one point, and losses earn zero points. It has many deficiencies from a theoretical point of view, for instance, reversing all results does not necessarily reverse the ranking and the number of total points allocated depends not only on the number of matches played. But teams should be definitely ranked even in these leagues.

Therefore, similarly to the \emph{quotient rule} used in the unfinished 2019/20 Handball-Bundesliga to determine the final ranking, the average score per game is calculated for each team by dividing its number of points by its number of matches played. These values are normalised such that the average quotient is subtracted from the quotient of each team to obtain $s_i$. Note that if two teams $i,j$ have played the same number of matches and $p_i = p_j$, then their normalised points are also equal, $s_i = s_j$.

\begin{example} \label{Examp1}
Eintracht Frankfurt has played $24$ matches and scored $p_i = 28$ points in the 2019/20 German Bundesliga until 13 March 2020. Therefore, it has scored $28 / 24 \approx 1.167$ points per game. The sum of these quotients for all the $18$ teams is $25$, thus the normalised point of Eintracht Frankfurt is $s_i = 28/24-25/18 \approx -0.222$. Hertha BSC has also scored $28$ points but over $25$ matches, hence this club has a slightly lower normalised point of $-0.269$.
\end{example}

Although normalisation of the scores considers the possibly different number of matches played by the teams, it still does not reflect the strength of the opponents.

\begin{table}[t]
  \centering
  \caption{The remaining matches of two teams in the 2019/20 German Bundesliga}
  \label{Table1}
  \rowcolors{1}{gray!20}{}
    \begin{tabularx}{\textwidth}{Lc m{1cm} Lc} \toprule \hiderowcolors
    \multicolumn{2}{c}{\textbf{Borussia Dortmund}} & & \multicolumn{2}{c}{\textbf{RB Lepzig}} \\ 
    Opponent & Points & & Opponent & Points \\ \bottomrule \showrowcolors
    Hertha BSC & 28    & & FC Augsburg & 27 \\
    Fortuna Düsseldorf & 22    & & Hertha BSC & 28 \\
    1899 Hoffenheim & 35    & & Borussia Dortmund & 51 \\
    RB Leipzig & 50    & & Fortuna Düsseldorf & 22 \\
    Mainz 05 & 26    & & SC Freiburg & 36 \\
    Bayern Munich & 55    & & 1899 Hoffenheim & 35 \\
    SC Paderborn & 16    & & 1.\ FC Köln & 32 \\
    Schalke 04 & 37    & & Mainz 05 & 26 \\
    VfL Wolfsburg & 36    & & SC Paderborn & 16 \\ \toprule
    \textbf{Sum} & \textbf{305} & &  \textbf{Sum} & \textbf{273} \\ \bottomrule
    \end{tabularx}
\end{table}

\begin{example} \label{Examp2}
Table~\ref{Table1} compares the teams against which Borussia Dortmund and RB Leipzig should have played in the remaining nine rounds of the 2019/20 German Bundesliga. Dortmund had clearly a more difficult schedule ahead than Leipzig as the points scored by its future opponents in the previous rounds were higher by more than 11.7\% (all of these teams played 25 matches). However, Dortmund had a lead of one point over Leipzig. Is it sufficient to rank Dortmund higher than Leipzig?
\end{example}

The strength of the schedule can be taken into account through the matches matrix $\textbf{M}$. The matches already played can be represented by an undirected graph, where the nodes correspond to the teams and the weight of any edge is determined by the number of matches played by the associated teams. The \emph{Laplacian matrix} $\mathbf{L} = \left[ \ell_{ij} \right]$ of this graph is defined such that $\ell_{ij} = -m_{ij}$ for all $i \neq j$ and $\ell_{ii}$ is equal to the number of matches played by team $i$.
Furthermore, let $\mathbf{I}$ be the matrix with ones in the diagonal and zeros otherwise, and $\mathbf{e}$ be the vector with $e_i = 1$ for all $i$.

\begin{definition}
\emph{Generalised row sum}: the generalised row sum rating vector $\mathbf{x}(\varepsilon)$ is given by the unique solution of the system of linear equations
\[
\left[ \mathbf{I} + \varepsilon \mathbf{L} \right] \mathbf{x}(\varepsilon) = \mathbf{s},
\]
where $\varepsilon > 0$ is a parameter.
\end{definition}

The generalised row sum method was introduced in \citet{Chebotarev1989_eng} and \citet{Chebotarev1994}. It adjusts the normalised points by considering the performance of the opponents, the opponents of the opponents, and so on.
Parameter $\varepsilon$ quantifies the degree of this modification. The ranking induced by $\mathbf{x}(\varepsilon)$ converges to the ranking from the number of points as $\varepsilon \to 0$, hence the normalised point is a limit of generalised row sum.
Uniqueness comes from the fact that the matrix $\mathbf{I} + \varepsilon \mathbf{L}$ is invertible.
\citet[Section~4]{Chebotarev1994} demonstrates that $\mathbf{x}(\varepsilon)$ can be regarded as the prediction of the number of points at the end of the season after all omitted comparisons.

\begin{definition}
\emph{Least squares}: the least squares rating vector $\mathbf{q}$ is given by the solution of the system of linear equations $\mathbf{e}^\top \mathbf{q} = 0$ and 
\[
\mathbf{L} \mathbf{q} = \mathbf{s}.
\]
\end{definition}

The least squares rating is unique if and only if any team can be compared with any other team at least indirectly, through other teams. This is a natural constraint---otherwise, there exist two subsets of the teams without any matches between the two sets. The condition certainly holds if half of the season is finished, that is, each team has played all other teams at least once.
The method is called least squares since the above system of linear equations can be obtained as the optimal solution of a least squared errors estimation \citep{Gonzalez-DiazHendrickxLohmann2014, Csato2015a}.

The least squares rating vector $\mathbf{q}$ can be calculated recursively unless graph $G$ is regular bipartite \citep[Theorem~2]{Csato2015a}:
$\mathbf{q} = \lim_{k \to \infty} \mathbf{q}^{(k)}$, where
\begin{align} \label{eq_Buchholz_iteration}
\mathbf{q}^{(0)} & = (1/r) \mathbf{s}, \nonumber \\
\mathbf{q}^{(k)} & = \mathbf{q}^{(k-1)} + \frac{1}{r} \left[ \frac{1}{r} \left( r \mathbf{I} - \mathbf{L} \right) \right]^k \mathbf{s} \qquad \text{for all } k \geq 1,
\end{align}
In formula \eqref{eq_Buchholz_iteration}, $r$ is the maximal number of matches played by a team. Consequently, the entries of matrix $r \mathbf{I} - \mathbf{L}$ contain the number of matches played by the associated teams, assuming that a team has played against itself if it has played fewer matches than other teams.
Therefore, $\mathbf{q}^{(1)}$ modifies the normalised points due to the strength of the opponents, $\mathbf{q}^{(2)}$ accounts for the opponents of the opponents, and so on.

The ranking induced by $\mathbf{x}(\varepsilon)$ converges to the ranking from the least squares as $\varepsilon \to \infty$, hence the least squares is the other extremum of generalised row sum \citep[p.~326]{ChebotarevShamis1998a}.
The sum of points scored by the opponents is called the Buchholz point in chess and used as a tie-breaking criterion, justifying the alternative name \emph{recursive Buchholz} \citep{Brozos-VazquezCampo-CabanaDiaz-RamosGonzalez-Diaz2010}.

\begin{proposition}
The generalised row sum and least squares rankings satisfy self-consistency, invariance to cycles, home-away independence, and consistency.
\end{proposition}

\begin{proof}
\citet[Theorem~5]{ChebotarevShamis1998a} proves self-consistency.
Invariance to cycles holds because reversing all results along a cycle does not affect the points vector $\mathbf{p}$.
Home-away independence is guaranteed by disregarding the field of the game in the points vector $\mathbf{p}$ and the matches matrix $\textbf{M}$.
Consistency is verified for the generalised row sum by \citet[Property~3]{Chebotarev1994} and for the least squares by \citet[Proposition 5.3]{Gonzalez-DiazHendrickxLohmann2014}.
\end{proof}

In addition, these rankings are independent of the particular value $\alpha$ for wins in sports using a symmetric scoring rule. Analogously, the result is the same for the 3-1-0, 6-2-0, and 4-2-1 point systems, in other words, the rewards can be shifted and multiplied arbitrarily.

The generalised row sum and especially the least squares methods have several successful applications, including, among others, 
(a) international price comparisons by the Eurostat and OECD, where the least squares procedure is called the EKS (\'Eltet\H{o}-K\"oves-Szulc) method \citep{EurostatOECD2012};
(b) evaluating movies \citep{JiangLimYaoYe2011};
(c) deriving alternative quality of life indices \citep{Petroczy2021a};
(d) ranking the participants of the Eurovision Song Contest \citep{CaklovicKurdija2017};
(e) ranking universities on the basis of applicants' preferences \citep{CsatoToth2020};
(f) ranking historical go \citep{ChaoKouLiPeng2018} and tennis players \citep{BozokiCsatoTemesi2016}; and
(g) ranking teams in Swiss-system chess tournaments \citep{Csato2017c}.
Further details on their theoretical properties can be found in \citet{Chebotarev1994}, \citet{Shamis1994}, and \citet{Gonzalez-DiazHendrickxLohmann2014}.

\section{Application: ranking in the 2019/20 season of six major European soccer leagues} \label{Sec4}

The five major premier European soccer leagues have been chosen to illustrate how this parametric family of ranking methods work, together with the league in the Netherlands, which has been declared void on 25 April 2020. All of them have been suspended in March 2020 due to the coronavirus pandemic, when about 70-75\% of the matches in the 2019/20 season have already been played:
\begin{itemize}
\item
The \href{https://en.wikipedia.org/wiki/2019\%E2\%80\%9320_Premier_League}{Premier League} in England with 20 teams, 29 rounds finished except for four clubs, which have played only 28 games;
\item
The \href{https://en.wikipedia.org/wiki/2019\%E2\%80\%9320_Ligue_1}{Ligue 1} in France with 20 teams, 28 rounds finished except for the game Strasbourg vs.\ Paris Saint-Germain;
\item
The \href{https://en.wikipedia.org/wiki/2019\%E2\%80\%9320_Bundesliga}{Bundesliga} in Germany with 18 teams, 25 rounds finished except for the game Werder Bremen vs.\ Eintracht Frankfurt;
\item
The \href{https://en.wikipedia.org/wiki/2019\%E2\%80\%9320_Serie_A}{Seria A} in Italy with 20 teams, 26 rounds finished except for eight clubs, which have played only 25 games;
\item
The \href{https://en.wikipedia.org/wiki/2019\%E2\%80\%9320_Eredivisie}{Eredivisie} in the Netherlands with 18 teams, 26 rounds finished except for four clubs, which have played only 25 games;
\item
The \href{https://en.wikipedia.org/wiki/2019\%E2\%80\%9320_La_Liga}{La Liga} in Spain with 20 teams, 27 rounds finished.
\end{itemize}

\begin{table}[t]
  \centering
  \caption{Rankings of suspended European soccer leagues I. \\ \vspace{0.25cm}
  \footnotesize{Pts stands for the number of points. \\
  Teams with \textbf{bold} points have played one match fewer. \\
  Tie-breaking rules: fewer matches played, goal difference, goals scored. The generalised row sum and least squares methods do not require tie-breaking.}}
  \label{Table2}
  \rowcolors{1}{}{gray!20}
      \begin{tabularx}{\textwidth}{cccc C cccc C cccc} \toprule \hiderowcolors
    \multicolumn{4}{c}{\textbf{Italy}} &       & \multicolumn{4}{c}{\textbf{Netherlands}} &  &     \multicolumn{4}{c}{\textbf{Spain}} \\
    Pts   & \multicolumn{3}{c}{Value of $\varepsilon$} &       & Pts   & \multicolumn{3}{c}{Value of $\varepsilon$} &       & Pts   & \multicolumn{3}{c}{Value of $\varepsilon$} \\
          & 0.001 & 0.1   & $\to \infty$ &       &       & 0.001 & 0.1   & $\to \infty$       &        &       & 0.001 & 0.1   & $\to \infty$ \\ \bottomrule \showrowcolors
    63    & 1     & 1     & 1     &       & \textbf{56} & 1     & 1     & 1     &       & 58    & 1     & 1     & 1 \\
    62    & 2     & 2     & 2     &       & \textbf{56} & 2     & 2     & 2     &       & 56    & 2     & 2     & 2 \\
    \textbf{54} & 3     & 3     & 3     &       & \textbf{50} & 3     & 3     & 3     &       & 47    & 3     & 3     & 3 \\
    \textbf{48} & 4     & 4     & 4     &       & 49    & 4     & 4     & 4     &       & 46    & 5     & 5     & 5 \\
    45    & 5     & 5     & 5     &       & 44    & 5     & 5     & 5     &       & 46    & 4     & 4     & 4 \\
    39    & 6     & 6     & 6     &       & \textbf{41} & 6     & 6     & 6     &       & 45    & 6     & 6     & 6 \\
    36    & 9     & 9     & 9     &       & 41    & 7     & 7     & 7     &       & 42    & 7     & 7     & 7 \\
    \textbf{35} & 8     & 8     & 8     &       & 36    & 8     & 8     & 9     &       & 38    & 9     & 9     & 9 \\
    \textbf{35} & 7     & 7     & 7     &       & 35    & 9     & 9     & 8     &       & 38    & 8     & 8     & 8 \\
    34    & 10    & 10    & 11    &       & 33    & 10    & 10    & 10    &       & 37    & 10    & 10    & 10 \\
    \textbf{32} & 12    & 12    & 12    &       & 33    & 11    & 11    & 11    &       & 34    & 11    & 11    & 11 \\
    \textbf{32} & 11    & 11    & 10    &       & 32    & 12    & 12    & 12    &       & 33    & 12    & 12    & 12 \\
    30    & 13    & 13    & 13    &       & 28    & 13    & 13    & 15    &       & 33    & 13    & 13    & 13 \\
    28    & 15    & 14    & 14    &       & 27    & 14    & 14    & 13    &       & 32    & 14    & 14    & 14 \\
    \textbf{27} & 14    & 15    & 15    &       & 26    & 15    & 15    & 14    &       & 29    & 15    & 15    & 15 \\
    \textbf{26} & 16    & 16    & 16    &       & 26    & 16    & 16    & 16    &       & 27    & 16    & 16    & 16 \\
    25    & 17    & 17    & 17    &       & 19    & 17    & 17    & 17    &       & 26    & 17    & 17    & 17 \\
    25    & 18    & 18    & 18    &       & 15    & 18    & 18    & 18    &       & 25    & 18    & 18    & 18 \\
    18    & 19    & 19    & 19    &       &       &       &       &       &       & 23    & 19    & 19    & 19 \\
    16    & 20    & 20    & 20    &       &       &       &       &       &       & 20    & 20    & 20    & 20 \\ \bottomrule
    \end{tabularx}
\end{table}

\begin{table}[t]
  \centering
  \caption{Rankings of suspended European soccer leagues II. \\ \vspace{0.25cm}
  \footnotesize{Pts stands for the number of points. \\
  Teams with \textbf{bold} points have played one match fewer. \\
  Tie-breaking rules: fewer matches played, goal difference, goals scored. The generalised row sum and least squares methods do not require tie-breaking.}}
  \label{Table3}
  \rowcolors{1}{}{gray!20}
    \begin{tabularx}{\textwidth}{cccc C cccc C cccc} \toprule \hiderowcolors
    \multicolumn{4}{c}{\textbf{England}}   &       & \multicolumn{4}{c}{\textbf{France}}    &       & \multicolumn{4}{c}{\textbf{Germany}}   \\
    Pts   & \multicolumn{3}{c}{Value of $\varepsilon$} &       & Pts   & \multicolumn{3}{c}{Value of $\varepsilon$} &       & Pts   & \multicolumn{3}{c}{Value of $\varepsilon$} \\
          & 0.001 & 0.1   & $\to \infty$ &       &       & 0.001 & 0.1   & $\to \infty$ &       &       & 0.001 & 0.1   & $\to \infty$ \\ \bottomrule \showrowcolors
    82    & 1     & 1     & 1     &       & \textbf{68} & 1     & 1     & 1     &       & 55    & 1     & 1     & 1 \\
    \textbf{57} & 2     & 2     & 2     &       & 56    & 2     & 2     & 2     &       & 51    & 2     & 2     & 3 \\
    53    & 3     & 3     & 3     &       & 50    & 3     & 3     & 4     &       & 50    & 3     & 3     & 2 \\
    48    & 4     & 4     & 4     &       & 49    & 4     & 4     & 3     &       & 49    & 4     & 4     & 4 \\
    45    & 5     & 5     & 5     &       & 41    & 5     & 5     & 5     &       & 47    & 5     & 5     & 5 \\
    \textbf{43} & 6     & 6     & 6     &       & 41    & 6     & 6     & 8     &       & 37    & 6     & 6     & 6 \\
    43    & 7     & 7     & 7     &       & 40    & 9     & 9     & 9     &       & 36    & 8     & 9     & 9 \\
    41    & 9     & 8     & 8     &       & 40    & 7     & 7     & 6     &       & 36    & 7     & 8     & 8 \\
    \textbf{40} & 8     & 9     & 9     &       & 40    & 8     & 8     & 7     &       & 35    & 9     & 7     & 7 \\
    39    & 10    & 10    & 10    &       & 39    & 11    & 10    & 10    &       & 32    & 10    & 10    & 10 \\
    39    & 11    & 11    & 11    &       & \textbf{38} & 10    & 11    & 12    &       & 30    & 11    & 11    & 11 \\
    37    & 12    & 12    & 12    &       & 37    & 13    & 13    & 13    &       & \textbf{28} & 12    & 12    & 12 \\
    35    & 13    & 13    & 13    &       & 37    & 12    & 12    & 11    &       & 28    & 13    & 14    & 14 \\
    34    & 14    & 14    & 14    &       & 34    & 15    & 15    & 15    &       & 27    & 14    & 13    & 13 \\
    29    & 15    & 16    & 16    &       & 34    & 14    & 14    & 14    &       & 26    & 15    & 15    & 15 \\
    27    & 16    & 15    & 15    &       & 30    & 16    & 16    & 16    &       & 22    & 16    & 16    & 16 \\
    27    & 17    & 17    & 17    &       & 30    & 17    & 17    & 17    &       & \textbf{18} & 17    & 17    & 17 \\
    27    & 18    & 18    & 18    &       & 27    & 18    & 18    & 18    &       & 16    & 18    & 18    & 18 \\
    \textbf{25} & 19    & 19    & 19    &       & 23    & 19    & 19    & 19    &       &       &       &       &  \\
    21    & 20    & 20    & 20    &       & 13    & 20    & 20    & 20    &       &       &       &       &  \\ \bottomrule
    \end{tabularx}
\end{table}

Tables~\ref{Table2} and \ref{Table3} report the current number of points for each team, their rankings by the generalised row sum method with two values of $\varepsilon$, as well as by the least squares method. In the case of the smallest parameter, $\varepsilon = 0.001$, the strength of opponents serves only to break the ties in the normalised points.
%While it would be difficult to prove that the same rank for these particular values of $\varepsilon$ extends to all possible choices, it is very likely.

The first and most crucial observation is the robustness of the rankings. This stands in stark contrast to Swiss-system chess team tournaments, where accounting for the opponents can substantially affect the ranking, particularly for the middle teams \citep{Csato2017c}. The probable reason is the large number of matches already played. It is a favorable finding, which indicates that freezing the current standing---with the tie-breaking rules according to Tables~\ref{Table2} and \ref{Table3}---or using the quotient rule (normalised points) would be a relatively fair solution.

The ideal case is represented by Spain: if the ties in the number of points are resolved through the strength of opponents ($\varepsilon \to 0$), then the ranking does not depend on the weight of the adjustment at all. There is no need to consider the number of matches played in La Liga.
Compared to the ranking induced by the number of points, the only change is on the fourth position as officially Real Sociedad overtakes Getafe due to more goals scored, but the latter team has faced stronger opponents. This difference has fundamental sporting effects since the fourth club automatically qualifies for the group stage of the UEFA Champions League, and the fifth goes only to the UEFA Europa League.

While some top clubs in the Netherlands have played fewer matches, this does not influence the usual ranking with goal difference as the tie-breaking rule. Crucially, the parameter responsible for adjusting the normalised points due to the strength of the opponents does not change the ranking except for certain middle teams, which has not much sporting effect. This is important because Ajax Amsterdam---the team having the most ($56$) points in a tie with AZ Alkmaar, favoured by goal difference---has not been declared the champion but it has obtained the only slot for the Netherlands in the group stage of the 2020/21 UEFA Champions League, while the second club has entered the second qualifying round in the league path, therefore it should have won three clashes to reach the group stage. Interestingly, AZ Alkmaar is preferred to Ajax by the eigenvector-based method of \citet{LambersSpieksma2020b}. However, the remaining opponents of Ajax scored 321 points, while the remaining opponents of AZ Alkmaar scored 335 points, which supports the conjecture that Ajax had a more difficult schedule in the matches already played.

Similarly, the value of $\varepsilon$ barely affects the ranking in Italy. Disregarding tie-breaking, the changes are on the 10th and the 14th positions, which are almost insignificant because they influence neither qualification for European cups nor relegation. On the other hand, the seventh place---which provides a slot in the second qualifying round of the UEFA Europa League---should be given to one of the teams that scored 35 points, Hellas Verona or Parma as the additional match played by Milan is seemingly responsible for its advantage of one point. Although Parma has a worse goal difference, this club has faced stronger opponents than Hellas Verona.

The situation in the Premier League is analogous to Serie A. While Arsenal has one point less ($40$) compared to Tottenham Hotspur, this is compensated by the fewer number of games played for small values of $\varepsilon$. The generalised row sum method also exchanges Brighton \& Hove Albion with West Ham United on the insignificant 15th place.

Turning to the Bundesliga, the $13$th place has minimal sporting effects. The seventh team qualifies for the second qualifying round of the UEFA Europa League because the cup winner, Bayern Munich, is one of the first six clubs. Furthermore, as Example~\ref{Examp2} has explained, RB Leipzig can be the runner-up instead of Borussia Dortmund if the strength of the opponents receives a considerable role in the ranking, that is, for large values of $\varepsilon$. Nonetheless, both the second- and the third-placed clubs enter the group stage of the UEFA Champions League.

The ranking in Ligue 1---a league that has never been finished---is perhaps the most challenging. Since Paris Saint-Germain (first with 68 points) has won both the \href{https://en.wikipedia.org/wiki/2019\%E2\%80\%9320_Coupe_de_France}{domestic cup} and the \href{https://en.wikipedia.org/wiki/2019\%E2\%80\%9320_Coupe_de_la_Ligue}{French league cup competition} and the UEFA Europa League titleholder, the Spanish Sevilla has qualified for the UEFA Champions League through its domestic league, the third-placed club has qualified for the Champions League group stage, the fourth- and the fifth-placed clubs have gone to the group stage of the Europa League, and the sixth-placed club has received a slot in the second qualifying round of the Europa League.

If the strength of the opponents takes a large weight, then Lille with 49 points is ranked above Rennes with 50 points and obtains the Champions League slot, while Montpellier with 40 points is ranked above Nice with 41 points and goes to the Europa League. The issue of Rennes and Lille has been highlighted in \citet{Guyon2020b}: Lille played and lost twice against the dominating Paris Saint-Germain while Rennes faced it once, and the number of points scored so far by their future opponents is 347 for Lille but 379 for Rennes (and even 381 or 382 if accounting for the missing clash Strasbourg vs.\ Paris Saint-Germain from the 29th round).

Finally, recall that the least squares rating vector can be computed recursively according to formula \eqref{eq_Buchholz_iteration}. Accounting for the strength of direct opponents, that is, vector $\mathbf{q}^{(1)}$ almost always induces the same ranking as the final rating vector $\mathbf{q}$, the only exception being that Lyon with 40 points in Ligue 1 is ranked above Angers having 39 points only after considering the opponents of opponents (in addition, the ranking according to $\mathbf{q}^{(2)}$ is different from the ranking induced by $\mathbf{q}^{(1)}$, $\mathbf{q}^{(3)}$, and $\mathbf{q}$ in the 2019/20 Premier League).
In our examples, primarily the direct impact of the opponents count. Consequently, it is not necessary to calculate the inverse of any matrix, and one matrix multiplication can be sufficient to determine a fair ranking. 

\section{Policy implications} \label{Sec5}

Ranking in a suspended round-robin league is a non-trivial problem because of the need to account for schedule imbalances and possibly for the different number of matches played by the teams. The quotient rule addresses the second issue but does not reflect the strength of opponents. Furthermore, if most teams have played the same number of matches in an unfinished season, which is the likely case, then tie-breaker criteria are required to be used almost as often as in any round-robin league to assign unique ranks to competitors that are equal on points.

However, both goal difference (winning margin) and head-to-head records are dubious rules in the presence of unplayed matches: the former favours teams that have played against weaker teams, and the latter is strongly influenced by the home-away pattern of the games played. For instance, see the tie-breaking decisions in the 2019/20 season of the French Ligue 1 \citep{Ligue1_2020}. Nice was fifth ahead of Reims due to better head-to-head results over their home and away matches, although Reims had a higher goal difference. Brest was ranked above Metz because of superior goal difference as they played only one match. Since the original tie-breaking rule at the end of the season would have been goal difference, the choice of the French top soccer league seems to be somewhat arbitrary.

Regarding the proposed generalised row sum method, even though it is difficult to argue for any particular value of $\varepsilon$, we think it is best to give the largest possible weight to the role of opponents. At least, there is some evidence for supporting the least squares ranking in Swiss-system chess team tournaments \citep{Csato2017c}. Since this procedure does not contain a parameter, it would be a better alternative than the generalised row sum, which may lead to such disputes.

Therefore it is recommended to choose from the following solutions, listed in decreasing order of complexity and preference:
\begin{enumerate}
\item
Least squares method;
\item
Generalised row sum method with a small $\varepsilon$ that breaks only the ties remaining in the normalised points;
\item
Quotient rule (points per game) with the tie-breaking criteria of goal difference, goals scored;
\item
Number of points with the tie-breaking criteria of fewer matches played, goal difference, goals scored.
\end{enumerate}

\section{Conclusions} \label{Sec6}

The problem of ranking the teams in an incomplete round-robin tournament has been discussed. We have taken an axiomatic approach to identify the key properties of real-world rankings. A family of ranking methods has been shown to satisfy the desired requirements. It has been applied for the 2019/20 season of six major European soccer leagues that were suspended after the outbreak of coronavirus.

Accounting for the strength of opponents turns out to be necessary for incomplete tournaments. This factor can be considered only by introducing an additional parameter. However, its value has only marginal influence in all leagues considered here as the set of clubs to be relegated is entirely independent of this choice, and qualification to the European cups is barely influenced.

The adjustment of the standard scores is carried out in two steps: the correction for the number of matches played (normalised points) is followed by accounting for the strength of the opponents. This provides a simple, robust, and scientifically well-established way to obtain a fair final standing in a round-robin league that cannot resume.

Naturally, ranking in a sports league is a zero-sum game, thus any solution will prefer certain teams compared to an alternative regime. Future rulebooks should explicitly define what happens if a league has to be finished without playing all matches. Our paper can contribute to single out a suitable policy for this purpose.

\section*{Acknowledgements}
\addcontentsline{toc}{section}{Acknowledgements}
\noindent
My father (also called \emph{L\'aszl\'o Csat\'o}) helped in executing the numerical calculations. \\
We are deeply indebted to \emph{Julien Guyon} for inspiration and \emph{Manel Baucells} for useful comments. \\
%We are grateful to \emph{Tam\'as Halm} for useful advice. \\
Four anonymous reviewers provided valuable comments and suggestions on an earlier draft. \\
We are grateful to the \href{https://en.wikipedia.org/wiki/Wikipedia_community}{Wikipedia community} for collecting and structuring invaluable information on the sports tournaments discussed. We found---and corrected---two mistakes (Nantes vs.\ Angers was 1-2 instead of 1-0 and Monaco vs.\ Nimes was not annulled but 2-2) in the \href{https://en.wikipedia.org/wiki/2019\%E2\%80\%9320_Ligue_1\#Results}{table summarising the results of the 2019/20 French Ligue 1} on 29 April 2020. \\
The research was supported by the MTA Premium Postdoctoral Research Program grant PPD2019-9/2019.

\bibliographystyle{apalike}
\bibliography{All_references}

\begin{thebibliography}{}

\bibitem[{BBC}, 2020]{BBC2020a}
{BBC} (2020).
\newblock Ligue 1 \& 2: {F}rance's top two divisions will not resume this
  season.
\newblock 28 April. \url{https://www.bbc.com/sport/football/52460468}.

\bibitem[Beggs et~al., 2019]{BeggsBondEmmondsJones2019}
Beggs, C.~B., Bond, A.~J., Emmonds, S., and Jones, B. (2019).
\newblock Hidden dynamics of soccer leagues: The predictive `power' of partial
  standings.
\newblock {\em Plos One}, 14(12):e0225696.

\bibitem[Berganti{\~n}os and Moreno-Ternero,
  2020]{BergantinosMoreno-Ternero2020a}
Berganti{\~n}os, G. and Moreno-Ternero, J.~D. (2020).
\newblock Sharing the revenues from broadcasting sport events.
\newblock {\em Management Science}, 66(6):2417--2431.

\bibitem[Berker, 2014]{Berker2014}
Berker, Y. (2014).
\newblock Tie-breaking in round-robin soccer tournaments and its influence on
  the autonomy of relative rankings: {UEFA} vs.\ {FIFA} regulations.
\newblock {\em European Sport Management Quarterly}, 14(2):194--210.

\bibitem[Boz{\'o}ki et~al., 2016]{BozokiCsatoTemesi2016}
Boz{\'o}ki, S., Csat{\'o}, L., and Temesi, J. (2016).
\newblock An application of incomplete pairwise comparison matrices for ranking
  top tennis players.
\newblock {\em European Journal of Operational Research}, 248(1):211--218.

\bibitem[Brams and Ismail, 2018]{BramsIsmail2018}
Brams, S.~J. and Ismail, M.~S. (2018).
\newblock Making the rules of sports fairer.
\newblock {\em SIAM Review}, 60(1):181--202.

\bibitem[Brozos-V\'azquez et~al.,
  2010]{Brozos-VazquezCampo-CabanaDiaz-RamosGonzalez-Diaz2010}
Brozos-V\'azquez, M., Campo-Cabana, M.~A., D\'iaz-Ramos, J.~C., and
  Gonz\'alez-D\'iaz, J. (2010).
\newblock Recursive tie-breaks for chess tournaments.
\newblock
  \url{http://eio.usc.es/pub/julio/Desempate/Performance_Recursiva_en.htm}.

\bibitem[{\v{C}}aklovi{\'c} and Kurdija, 2017]{CaklovicKurdija2017}
{\v{C}}aklovi{\'c}, L. and Kurdija, A.~S. (2017).
\newblock A universal voting system based on the {P}otential {M}ethod.
\newblock {\em European Journal of Operational Research}, 259(2):677--688.

\bibitem[Chao et~al., 2018]{ChaoKouLiPeng2018}
Chao, X., Kou, G., Li, T., and Peng, Y. (2018).
\newblock Jie {K}e versus {A}lpha{G}o: A ranking approach using decision making
  method for large-scale data with incomplete information.
\newblock {\em European Journal of Operational Research}, 265(1):239--247.

\bibitem[Chebotarev, 1989]{Chebotarev1989_eng}
Chebotarev, P. (1989).
\newblock Generalization of the row sum method for incomplete paired
  comparisons.
\newblock {\em Automation and Remote Control}, 50(8):1103--1113.

\bibitem[Chebotarev and Shamis, 1997]{ChebotarevShamis1997a}
Chebotarev, P. and Shamis, E. (1997).
\newblock Constructing an objective function for aggregating incomplete
  preferences.
\newblock In Tangian, A. and Gruber, J., editors, {\em Constructing
  Scalar-Valued Objective Functions}, volume 453 of {\em Lecture Notes in
  Economics and Mathematical Systems}, pages 100--124. Springer,
  Berlin-Heidelberg.

\bibitem[Chebotarev, 1994]{Chebotarev1994}
Chebotarev, P.~{\relax Yu}. (1994).
\newblock Aggregation of preferences by the generalized row sum method.
\newblock {\em Mathematical Social Sciences}, 27(3):293--320.

\bibitem[Chebotarev and Shamis, 1998]{ChebotarevShamis1998a}
Chebotarev, P.~{\relax Yu}. and Shamis, E. (1998).
\newblock Characterizations of scoring methods for preference aggregation.
\newblock {\em Annals of Operations Research}, 80:299--332.

\bibitem[Csat\'o, 2015]{Csato2015a}
Csat\'o, L. (2015).
\newblock A graph interpretation of the least squares ranking method.
\newblock {\em Social Choice and Welfare}, 44(1):51--69.

\bibitem[Csat\'o, 2017]{Csato2017c}
Csat\'o, L. (2017).
\newblock On the ranking of a {S}wiss system chess team tournament.
\newblock {\em Annals of Operations Research}, 254(1-2):17--36.

\bibitem[Csat\'o, 2019]{Csato2019d}
Csat\'o, L. (2019).
\newblock An impossibility theorem for paired comparisons.
\newblock {\em Central European Journal of Operations Research},
  27(2):497--514.

\bibitem[Csat\'o, 2021a]{Csato2021c}
Csat\'o, L. (2021a).
\newblock A comparison of penalty shootout designs in soccer.
\newblock {\em 4OR}, 19:183--198.

\bibitem[Csat\'o, 2021b]{Csato2021a}
Csat\'o, L. (2021b).
\newblock {\em Tournament Design: How Operations Research Can Improve Sports
  Rules}.
\newblock Palgrave Pivots in Sports Economics. Palgrave Macmillan, Cham,
  Switzerland.

\bibitem[Csat\'o and T\'oth, 2020]{CsatoToth2020}
Csat\'o, L. and T\'oth, {\relax Cs}. (2020).
\newblock University rankings from the revealed preferences of the applicants.
\newblock {\em European Journal of Operational Research}, 286(1):309--320.

\bibitem[Duckworth and Lewis, 1998]{DuckworthLewis1998}
Duckworth, F.~C. and Lewis, A.~J. (1998).
\newblock A fair method for resetting the target in interrupted one-day cricket
  matches.
\newblock {\em Journal of the Operational Research Society}, 49(3):220--227.

\bibitem[{European Union / OECD}, 2012]{EurostatOECD2012}
{European Union / OECD} (2012).
\newblock {\em Eurostat-OECD Methodological Manual on Purchasing Power
  Parities}.
\newblock Publications Office of the European Union, Luxembourg.

\bibitem[Gonz\'alez-D\'iaz et~al., 2014]{Gonzalez-DiazHendrickxLohmann2014}
Gonz\'alez-D\'iaz, J., Hendrickx, R., and Lohmann, E. (2014).
\newblock Paired comparisons analysis: an axiomatic approach to ranking
  methods.
\newblock {\em Social Choice and Welfare}, 42(1):139--169.

\bibitem[Gorgi et~al., 2020]{GorgiKoopmanLit2020}
Gorgi, P., Koopman, S.~J., and Lit, R. (2020).
\newblock Estimation of final standings in football competitions with premature
  ending: the case of {COVID}-19.
\newblock Manuscript.
  \url{https://www.timeserieslab.com/articles/football.pdf}.

\bibitem[Guyon, 2020a]{Guyon2020b}
Guyon, J. (2020a).
\newblock Football : comment d\'ecider du classement final de la {L}igue 1 si
  elle devait s'arr{\^e}ter ici ?
\newblock {\em Le Monde}.
\newblock 16 March.
  \url{https://www.lemonde.fr/sport/article/2020/03/16/football-comment-decider-du-classement-final-de-la-ligue-1-si-elle-devait-s-arreter-ici_6033217_3242.html}.

\bibitem[Guyon, 2020b]{Guyon2020c}
Guyon, J. (2020b).
\newblock The model to determine {P}remier {L}eague standings.
\newblock {\em The Times}.
\newblock 18 March.
  \url{https://www.thetimes.co.uk/article/the-model-to-determine-premier-league-standings-ttt8tnldd}.

\bibitem[Jiang et~al., 2011]{JiangLimYaoYe2011}
Jiang, X., Lim, L.-H., Yao, Y., and Ye, Y. (2011).
\newblock Statistical ranking and combinatorial {H}odge theory.
\newblock {\em Mathematical Programming}, 127(1):203--244.

\bibitem[Kitti, 2016]{Kitti2016}
Kitti, M. (2016).
\newblock Axioms for centrality scoring with principal eigenvectors.
\newblock {\em Social Choice and Welfare}, 46(3):639--653.

\bibitem[Lambers and Spieksma, 2020a]{LambersSpieksma2020b}
Lambers, R. and Spieksma, F. (2020a).
\newblock True rankings.
\newblock Manuscript.
  \url{https://www.euro-online.org/websites/orinsports/wp-content/uploads/sites/10/2020/05/TrueRanking.pdf}.

\bibitem[Lambers and Spieksma, 2020b]{LambersSpieksma2020a}
Lambers, R. and Spieksma, F.~C.~R. (2020b).
\newblock A mathematical analysis of fairness in shootouts.
\newblock {\em IMA Journal of Management Mathematics}, in press.
\newblock {DOI}:
  \href{https://doi.org/10.1093/imaman/dpaa023}{10.1093/imaman/dpaa023}.

\bibitem[Landau, 1914]{Landau1914}
Landau, E. (1914).
\newblock {\"U}ber {P}reisverteilung bei {S}pielturnieren.
\newblock {\em Zeitschrift f\"ur Mathematik und Physik}, 63:192--202.

\bibitem[Lenten, 2008]{Lenten2008}
Lenten, L.~J.~A. (2008).
\newblock Unbalanced schedules and the estimation of competitive balance in the
  {S}cottish {P}remier {L}eague.
\newblock {\em Scottish Journal of Political Economy}, 55(4):488--508.

\bibitem[{Ligue 1}, 2020]{Ligue1_2020}
{Ligue 1} (2020).
\newblock 2019/2020 {R}anking {T}able {R}ound 28.
\newblock \url{https://www.ligue1.com/ranking?seasonId=2019-2020&matchDay=28}.

\bibitem[Petr{\'o}czy, 2021]{Petroczy2021a}
Petr{\'o}czy, D.~G. (2021).
\newblock An alternative quality of life ranking on the basis of remittances.
\newblock {\em Socio-Economic Planning Sciences}, in press.
\newblock {DOI}:
  \href{https://doi.org/10.1016/j.seps.2021.101042}{10.1016/j.seps.2021.101042}.

\bibitem[Rubinstein, 1980]{Rubinstein1980}
Rubinstein, A. (1980).
\newblock Ranking the participants in a tournament.
\newblock {\em SIAM Journal on Applied Mathematics}, 38(1):108--111.

\bibitem[Shamis, 1994]{Shamis1994}
Shamis, E. (1994).
\newblock Graph-theoretic interpretation of the generalized row sum method.
\newblock {\em Mathematical Social Sciences}, 27(3):321--333.

\bibitem[{Sky Sports}, 2020]{SkySports2020}
{Sky Sports} (2020).
\newblock Ajax denied title as {D}utch {E}redivisie season declared void,
  {E}uropean places decided, no relegation.
\newblock 25 April.
  \url{https://www.skysports.com/football/news/11906/11978351/ajax-denied-title-as-dutch-eredivisie-season-declared-void-european-places-decided-no-relegation}.

\bibitem[Slutzki and Volij, 2005]{SlutzkiVolij2005}
Slutzki, G. and Volij, O. (2005).
\newblock Ranking participants in generalized tournaments.
\newblock {\em International Journal of Game Theory}, 33(2):255--270.

\bibitem[Slutzki and Volij, 2006]{SlutzkiVolij2006}
Slutzki, G. and Volij, O. (2006).
\newblock Scoring of web pages and tournaments -- axiomatizations.
\newblock {\em Social Choice and Welfare}, 26(1):75--92.

\bibitem[Stern, 2009]{Stern2009}
Stern, S.~E. (2009).
\newblock An adjusted {D}uckworth--{L}ewis target in shortened limited overs
  cricket matches.
\newblock {\em Journal of the Operational Research Society}, 60(2):236--251.

\bibitem[{Stregspiller}, 2020]{Stregspiller2020}
{Stregspiller} (2020).
\newblock It does not end here.
\newblock 11 April. \url{https://www.stregspiller.com/it-does-not-end-here/}.

\bibitem[Van~Eetvelde et~al., 2021]{VanEetveldeHvattumLey2021}
Van~Eetvelde, H., Hvattum, L.~M., and Ley, C. (2021).
\newblock The {P}robabilistic {F}inal {S}tanding {C}alculator: a fair
  stochastic tool to handle abruptly stopped football seasons.
\newblock Manuscript. arXiv:
  \href{https://arxiv.org/abs/2101.10597}{2101.10597}.

\bibitem[Vaziri et~al., 2018]{VaziriDabadghaoYihMorin2018}
Vaziri, B., Dabadghao, S., Yih, Y., and Morin, T.~L. (2018).
\newblock Properties of sports ranking methods.
\newblock {\em Journal of the Operational Research Society}, 69(5):776--787.

\bibitem[{Web24 News}, 2020]{Web24_2020}
{Web24 News} (2020).
\newblock Season in handball {B}undesliga canceled -- {K}iel champions.
\newblock 21 April.
  \url{https://www.web24.news/u/2020/04/season-in-handball-bundesliga-canceled-kiel-champions.html}.

\bibitem[Weisberg, 2014]{Weisberg2014}
Weisberg, H.~I. (2014).
\newblock {\em Willful Ignorance: The Mismeasure of Uncertainty}.
\newblock John Wiley \& Sons.

\bibitem[Wright, 2009]{Wright2009}
Wright, M. (2009).
\newblock 50 years of {OR} in sport.
\newblock {\em Journal of the Operational Research Society}, 60(Supplement
  1):S161--S168.

\end{thebibliography}

\end{document}